\documentclass[epj]{svjour}
\RequirePackage{float}
\usepackage{epstopdf}
\usepackage{hyperref}
\usepackage{amssymb}
\usepackage{mathtools}
\usepackage{ragged2e}
\usepackage{capt-of}
\usepackage{subfigure}
\usepackage{graphicx}  
\usepackage{dcolumn}   
\usepackage{bm}%
\usepackage{orcidlink}

\begin{document}


\title{ENERGY CONDITIONS FOR A $(WRS)_4$ SPACETIME IN $F(R)$-GRAVITY}
\author{Avik De\inst{1}\orcidlink{0000-0001-6475-3085}, Tee-How Loo\inst{2}\orcidlink{0000-0003-4099-9843}, Simran Arora\inst{3}\orcidlink{0000-0003-0326-8945}  \and P.K. Sahoo\inst{4}\orcidlink{0000-0003-2130-8832}
}                     
%
%
\institute{Department of Mathematical and Actuarial Sciences, Universiti Tunku Abdul Rahman, Jalan Sungai Long, 43000 Cheras, Malaysia,  Email: de.math@gmail.com
\and Institute of Mathematical Sciences, University of Malaya, 50603 Kuala Lumpur, Malaysia, Email: looth@um.edu.my
\and Department of Mathematics, Birla Institute of
Technology and Science-Pilani, Hyderabad Campus, Hyderabad-500078,
India, Email: dawrasimran27@gmail.com
\and Department of Mathematics, Birla Institute of
Technology and Science-Pilani, Hyderabad Campus, Hyderabad-500078,
India,  Email:  pksahoo@hyderabad.bits-pilani.ac.in
}
\date{Received: 21st Oct. 2020/ Accepted: 09th Feb 2021}
%
\abstract{
The objective of the present paper is to study 4-dimensional weakly Ricci symmetric spacetimes  $(WRS)_4$ 
with non-zero constant Ricci scalar. 
We prove that such a  $(WRS)_4$ satisfying $F(R)$-gravity field equations represents a perfect fluid with vanishing vorticity. 
Some energy conditions are studied under the current setting to constrain the functional form of $F(R)$. 
We examine a couple of popular toy models in $F(R)$-gravity, 
$F(R)=e^{\alpha R}$ where $\alpha$ is constant and $F(R)=R-\beta \tanh(R)$, 
$\beta$ is a constant. We also find that the equation of state parameter (EoS) in 
both models supports the universe's accelerating behavior, i.e., $\omega=-1$. 
According to the recently suggested observations of accelerated expansion, both cases define that the null, 
weak, and dominant energy conditions justify their requirements while the strong energy conditions violate them.
\PACS{
      {04.50.Kd}   \and
      {98.80.Es}{98.80.Cq}{98.80.-k}
     } 
} 
\titlerunning{ENERGY CONDITIONS FOR A $(WRS)_4$ SPACETIME IN $F(R)$-GRAVITY} 
\authorrunning{A. De et al.}
\maketitle

\section{{Introduction}}
	
Einstein's field equations (EFE) 
\begin{equation}
R_{ij}-\frac{R}{2}g_{ij}=\kappa^2 T_{ij},
\end{equation}
 where $\kappa^2=8\pi G$, $G$ being Newton's gravitational constant and $R=R^i_i$ the Ricci scalar, imply that the energy-momentum tensor $T_{ij}$ is of vanishing divergence. This requirement is satisfied if $T_{ij}$ is covariantly constant. Chaki and Ray showed that a general relativistic spacetime with covariant-constant energy-momentum tensor is Ricci symmetric, that is, $\nabla_iR_{jl} = 0$ \cite{chakiray}. Generalizing this concept, in \cite{tamsbin93} Tam\'assy and Binh introduced the notion of a weakly Ricci symmetric manifold $(WRS)_{n}$ as a non-flat Riemannian manifold of dimension $n> 2$ whose Ricci tensor $R_{ij}$ of type $(0,2)$ is not identically zero and satisfies the condition
\begin{equation} 
\nabla _{i}R_{jl}= A_iR_{jl}+ B_jR_{li}+ D_lR_{ij},\label{wrs}
\end{equation}
where $A_i,B_i,D_i$ are three non-zero 1-forms. 

General relativity models the universe as a four-dimensional smooth, connected, para-compact, Hausdorff spacetime manifold with a Lorentzian metric of signature $(-,+,+,+)$. A Lorentzian manifold is said to be a weakly Ricci symmetric spacetime if the Ricci tensor satisfies (\ref{wrs}).

A compact, orientable $(WRS)_n$ of constant Ricci scalar without boundary admitting a non-isometric conformal transformation was shown to be isometric to a sphere \cite{soochow}. The authors also obtained a sufficient condition for a compact, orientable $(WRS)_n$ without boundary to be conformal to a sphere in $E_{n+1}$. A conformally flat $(WRS)_4$ with non-zero Ricci scalar is proved to be a Robertson-Walker spacetime \cite{desahanous}. Shaikh and Kundu \cite{shaikh} proved some necessary and sufficient conditions for a warped product manifold to be $(WRS)_n$ and thus produced a condition for a Robertson-Walker spacetime to be $(WRS)_4$. Recently, Mantica and Molinari used Lovelock’s identity to discuss some general properties of $A_i$, $B_i$ and $D_i$ and proved that in a $(WRS)_n$, $B_i=D_i$ if the Ricci tensor is non-singular \cite{mantica}. A conformally flat perfect fluid $(WRS)_4$ spacetime obeying EFE without cosmological constant and having the basic vector field of $(WRS)_4$ as the velocity vector field of the fluid is infinitesimally spatially isotropic relative to the velocity vector field \cite{deghosh}. A non-Einstein quasi-conformally flat $(WRS)_n$ was shown to be $K$-special conformally flat and isometrically immersed in $E_{n+1}$ as a hypersurface \cite{shaikhacta}. De et al. (\cite{deghosh}, \cite{desahanous}) published several interesting results about curvature conditions in a $(WRS)_n$ spacetime considering $B_i\neq D_i$.  The first author \cite{avik} recently investigated  a $(WRS)_4$ spacetime satisfying EFE considering $B_i\neq D_i$. Several examples of $(WRS)_n$ are present in the literature. 
	
EFE are unable to explain the late time inflation of the universe without assuming the existence of some yet undetected components abbreviated as dark energy. This motivated some researchers to extend it to get some higher order field equations of gravity. One of these modified gravity theories is obtained by replacing the Ricci scalar $R$ in the Einstein-Hilbert action with an arbitrary function $F(R)$ of $R$. Of course the viability of such functions are constrained by several observational data and scalar-tensor theoretical results. Additionally we can always propose some phenomenological assumption about the form of the function $F(R)$ and later verify its validity from the present viability criteria. 

The matter content in EFE is more often assumed to be a perfect fluid continuum. 
Let $u_i$ denote a unit timelike vector. Then the spatial part $h_{ij}$ of the metric $g_{ij}$ can be defined as $h_{ij}=g_{ij}+u_iu_j$ so that $h_{ij}u^i=0$. $h^i_j$ thus can be called the projection operator orthogonal to the vector $u^i$. The energy momentum tensor $T_{ij}$ of type $(0,2)$ is given by $ T_{ij}=ph_{ij}+\rho u_iu_j,$ 
where $\rho=T_{ij}u^iu^j $ and $p=T_{ij}h^{ij}$ are the energy density and the isotropic pressure respectively, $u^i$ is called the four velocity vector of the fluid. The expansion scalar is given by $\theta=\theta_{ij}h^{ij}=\nabla_lu^l$. The shear is given by a symmetric tensor $s_{ij}=\theta_{ij}-\frac{1}{3}h_{ij}\theta$ and the vorticity by an anti-symmetric tensor $\Omega_{ij}=\nabla_{[_l}u_{k]}h^k_ih^l_j$. Both shear and vorticity tensors are orthogonal to $u^i$.

In addition, we assume that $p$ and $\rho $ are related by an equation of the form $p=\omega \rho$. 
Moreover, if $p=\rho$, then the perfect fluid is termed as stiff matter. 
The stiff matter era preceded the radiation era with $p=\frac{\rho}{3}$, 
the dust matter era with $p=0$ and followed by the dark matter era with $p=-\rho$ \cite{c2}. 
There are some works on energy conditions in different gravity models, 
in which the model parameters are often constrained by the equation of state parameter resulting in an accelerating universe (\cite{Sanjay}, \cite{Simran}).

The present paper is organized as follows: After the introduction, in Section 2 we study weakly Ricci symmetric spacetimes with constant Ricci scalar which satisfies $F(R)$-gravity equations. We show that the perfect fluid is vorticity free. In the next section we discuss the energy conditions in such a setting, followed by some toy models of $F(R)$-gravity, investigated in $(WRS)_4$ with constant $R$.  

\section{{$(WRS)_4$ satisfying $F(R)$-gravity}}

This section is devoted to study weakly Ricci symmetric spacetimes, so the covariant derivative of $R_{ij}$ satisfies (\ref{wrs}). If we denote $v_i=B_i-D_i$, (\ref{wrs}) gives us		
\begin{equation}
0=v_jR_i^l-R_{ji}v^l,\label{new}
\end{equation}
which on contraction over $i,j$ gives 
\begin{equation}
0=v^iR_i^l-Rv^l.\label{new1}
\end{equation}
Further, by transvecting with $v^j$ in (\ref{new}) we get
\begin{equation}
0=v^jv_jR_i^l-v^jR_{ji}v^l,
\end{equation}
which by (\ref{new1}) reduces to 
\begin{equation}
0=v^jv_jR_i^l-Rv_iv^l.
\end{equation}
Therefore, $R = 0$ if and only if  either $R_i^l = 0$ which is inadmissible by the definition of $(WRS)_4$ or $B_i-D_i=v_i=0$. So throughout the study we strictly assume that the spacetime is not scalar flat. Moreover, we can express $R_{ij}=-R v_i v_j$, where $v_i=B_i-D_i$ is considered to be unit timelike.

We consider a modified Einstein-Hilbert action term \cite{f(R)},
 \[S=\frac{1}{\kappa^2}\int F(R) \sqrt{-g}d^4x +\int L_m\sqrt{-g}d^4x,\]
where $F(R)$ is an arbitrary function of the Ricci scalar $R$, $L_m$ is the matter Lagrangian density, and we define the stress-energy tensor of matter as 
$$T_{ij}=-\frac{2}{\sqrt{-g}}\frac{\delta(\sqrt{-g}L_m)}{\delta g^{ij}}.$$ 

By varying the action $S$ of the gravitational field with respect to the metric tensor components $g^{ij}$ and using the least action principle we obtain the field equation
\begin{equation}
F_R(R)R_{ij}-\frac{1}{2}F(R)g_{ij}+(g_{ij}\Box-\nabla_i\nabla_j)F_R(R)=\kappa^2T_{ij},\label{FR}
\end{equation}
where $\Box$ represents the d'Alembertian operator, $F_R=\frac{\partial F(R)}{\partial R}$. Einstein's field equations can be reawakened by putting $F(R)=R$. For a constant Ricci scalar, we can express the above field equations (\ref{FR}) as follows:
\begin{equation}
 R_{ij}-\frac{R}{2}g_{ij}=\frac{\kappa^2}{F_R(R)}T^{\text{eff}}_{ij},\label{fr}
\end{equation}
where 
$$T^{\text{eff}}_{ij}=T_{ij}+T^{\text{curv}}_{ij}, \quad T^{\text{curv}}_{ij}=\frac{F(R)-RF_R(R)}{2\kappa^2}g_{ij}.$$ 
Remembering the term $\kappa^2=8\pi G$, the quantity $G^{\text{eff}}=\frac{G}{F_R(R)}$ can be regarded as the effective gravitational coupling strength in analogy to what is done in Brans-Dicke type scalar-tensor gravity theories and further the positivity of $G^{\text{eff}}$ (equivalent to the requirement that the graviton is not a ghost) imposes that the effective scalar degree of freedom or the scalaron term $f_R(R)> 0$.

If a $(WRS)_4$ satisfies (\ref{fr}), we can express 
\textbf{\begin{equation}
T^{\text{eff}}_{ij}=-\frac{RF_R(R)}{\kappa^2}v_iv_j-\frac{RF_R(R)}{2\kappa^2}g_{ij},\label{eff}
\end{equation}}
giving \textbf{$p^{\text{eff}}=p+\frac{F(R)-RF_R(R)}{2\kappa^2}$ and $\rho^{\text{eff}}=\rho-\frac{F(R)-RF_R(R)}{2\kappa^2}$.}

Thus beside the usual relation $T=\frac{RF_R(R)-2F(R)}{\kappa^2}$, between the trace of the energy momentum tensor and the geometry of the spacetime; under the present situation, this particular $T_{ij}$ denotes a perfect fluid type energy-momentum tensor whose isotropic pressure $p=-\frac{F(R)}{2\kappa^2}$ and density $\rho=\frac{F(R)-2RF_R(R)}{2\kappa^2}$ satisfy some specific relations with the geometry of the spacetime. This leads to our first result:\\

\begin{theorem}\label{pfthm}
In a $(WRS)_4$ with constant $R$ satisfying $F(R)$-gravity, the matter content is a perfect fluid with four-velocity vector $v_i$; constant isotropic pressure $p=-\frac{F(R)}{2\kappa^2}$ and constant energy density $\rho=\frac{F(R)-2RF_R(R)}{2\kappa^2}$.
\end{theorem}

\begin{theorem}The matter content in a $(WRS)_4$ spacetime with constant $R$ satisfying $F(R)$-gravity obeys the simple barotropic equation of state $p=\omega \rho$ if and only if $F(R)=\sigma R^\frac{1+\omega}{2\omega}$.
\end{theorem}
\begin{proof}
Suppose $p=\omega \rho$. Then from Theorem \ref{pfthm} it is quite straightforward that $(1+\omega)F(R)=2\omega RF_R(R)$ or $\frac{1+\omega}{2\omega}\frac{\partial R}{R}=\frac{\partial F(R)}{F(R)}$ and thus after integrating the equation we get the result, where $\sigma$ is an integrating constant.
\end{proof}
\begin{corollary}Corresponding to the different states of cosmic evolution of the universe we can conclude:
\begin{itemize}
\item The perfect fluid denotes dark matter $(\omega=-1)$ if $F(R)$ is a constant function of $R$ or alternately if the spacetime is scalar flat.
\item The perfect fluid denotes stiff matter $(\omega =1)$ if  $F(R)$ is a constant multiple of $R$. 
\item The perfect fluid denotes radiation $(\omega=1/3)$ if $F(R)$ is a constant multiple of $R^2$.
\item The perfect fluid cannot represent a dust era for any viable $F(R)$. 
\end{itemize}
\end{corollary}
\begin{theorem}
In a $(WRS)_4$ with constant $R$ satisfying $F(R)$-gravity, the matter content is a perfect fluid either with vanishing expansion scalar, acceleration vector and vorticity or represents a dark matter. 
\end{theorem}
\begin{proof}
Since $R$ is constant, the pressure and density of the perfect fluid as expressed earlier are constant too. 
Using the conservation of energy $g^{il}\nabla_lT_{ij}=0$ and  \eqref{eff}, we obtain
\begin{equation}
0=(p+\rho)\{\nabla_i v^i v_l+ v^i\nabla_i v_l\}, 
\end{equation}
Since $ v_l v^l=-1$, $\nabla_i v_l v^l=0$. 
Hence, either $p+\rho=0,$ or $\nabla_i v^i = v^i\nabla_i v_l=0$.

Since a conservative vector field is always irrotational, we get the vorticity of the perfect fluid is zero. 
\end{proof}

\begin{theorem}
In a perfect fluid spacetime with constant $R$ satisfying the $F(R)$-gravity, 
if the four velocity vector is given by $ v^i$, then its isotropic pressure and density are given respectively by 
$$p=-\frac{F(R)}{2\kappa^2}+\frac{F_R(R)}{3\kappa^2}\{R+ v^i v^jR_{ij}\}$$
and 
$$\rho=\frac{2v^i v^jR_{ij}F_R(R)+F(R)}{2\kappa^2}.$$
\end{theorem}

\begin{proof}
We have
\begin{equation}
F_R(R)R_{ij}-\frac{F(R)}{2}g_{ij}=\kappa^2(p+\rho) v_i v_j+\kappa^2 p g_{ij},\label{trho}
\end{equation}
which on contraction gives
\begin{equation}
R=\frac{3p-\rho}{F_R(R)}\kappa^2+2\frac{F(R)}{F_R(R)} \label{aa}
\end{equation}
On the other hand, by transvecting with $v^i v^j$  on  in (\ref{trho}), we get
\begin{equation}
v^i v^jR_{ij}=\frac{\kappa^2\rho-F(R)/2}{F_R(R)}.\label{e}
\end{equation}
From (\ref{aa}) and (\ref{e}) we obtain
\begin{equation}
F_R(R)v^i v^jR_{ij}+\frac{F(R)}{2}=\kappa^2 \rho=-RF_R(R)+3\kappa^2p+2F(R),
\end{equation}
which gives $3\kappa^2p=-\frac{3F(R)}{2}+F_R(R)\{v^i v^jR_{ij}+R\}$.
\end{proof}

\begin{theorem}
If a perfect fluid spacetime with constant $R$ satisfying $F(R)$-gravity obeys the timelike convergence condition, then $\rho\geq\frac{F(R)}{2\kappa^2}$, provided the four velocity vector is $ v^i$. 
\end{theorem}
\begin{proof}
$ v_i$ is a timelike vector, hence timelike convergence implies that $$v^i v^jR_{ij}\geq 0.$$ As discussed earlier, $F_R(R)$ is positive to ensure attractive gravity. 
Therefore, from (\ref{e}) we obtain the result. Equivalently, the result holds for $3\kappa^2p\geq RF_R(R)-\frac{3F(R)}{2}$. 
\end{proof}

\section{{Energy conditions in a $(WRS)_4$}}\label{conditions}

While exploring the possibility of different matter sources in the field equations of gravity, in general relativity and the extended theories of gravity, energy conditions come in handy to constrain the energy-momentum tensor and preserve the idea that not only the gravity is attractive but also the energy density is positive. In the case of a perfect fluid type effective matter in the $F(R)$-gravity theory, the conditions are given by:  

\begin{itemize}
\item Null energy condition \textbf{(NEC)}:\quad$\rho^{\text{eff}}+p^{\text{eff}}\geq 0$. 

\item Weak energy condition \textbf{(WEC)}: \quad $\rho^{\text{eff}}\geq 0$ and $\rho^{\text{eff}}+p^{\text{eff}}\geq 0$.

\item Strong energy condition \textbf{(SEC)}:\quad $\rho^{\text{eff}}+3p^{\text{eff}}\geq 0$ and $\rho^{\text{eff}}+p^{\text{eff}}\geq 0$.
\item Dominant energy condition \textbf{(DEC)}:\quad $\rho^{\text{eff}}\pm p^{\text{eff}}\geq 0$ and $\rho^{\text{eff}}\geq 0$. 
\end{itemize}
Using (\ref{eff}) and a non-negative $F(R)-RF_R(R)$, we can rewrite the conditions as the following \cite{Sahoo},

\begin{itemize}
\item Null energy condition \textbf{(NEC)}:\quad$\rho+p\geq 0$. In the present scenario, this gives us $RF_R(R)\leq 0$. 

\item Weak energy condition \textbf{(WEC)}: \quad $\rho\geq 0$ and $\rho+p\geq 0$. In the present context we obtain $F(R)-2RF_R(R)\geq 0$, together with $RF_R(R)\leq 0$

\item Strong energy condition \textbf{(SEC)}:\quad $\rho+3p\geq 0$ and $\rho+p\geq 0$. In the present context we obtain $RF_R(R)\leq 0$ when $F(R)+RF_R(R)\leq 0$.

\item Dominant energy condition \textbf{(DEC)}:\quad $\rho\pm p\geq 0$ and $\rho\geq 0$. In the present context we obtain $F(R)- 2RF_R(R)\geq 0$, together with $RF_R(R)\leq 0.$
\end{itemize} 

Since, $F_R(R)>0$ and $R\neq 0$, in our present model, $RF_R(R)\leq 0$ implies a negative Ricci scalar $R<0$.

\section{{Analysis of some toy models of $F(R)$-gravity in $(WRS)_4$}}
Many $F(R)$ models have been proposed in the literature. Here we consider few models of $F(R)$-gravity theories to analyse our results in a $(WRS)_4$ with constant Ricci scalar setting. 

\textbf{Case I:} $F(R)=\exp(\alpha R)$, $\alpha$ is a constant. There are some exponential models studied so far by S. I. Kruglov and S. D. Odintsov \cite{Kruglov,Odintsov}. The pressure and energy density for a perfect fluid continuum reads,
\begin{equation}
p=-\frac{e^{\alpha R}}{2 \kappa^{2}} \hspace{0.5 in} \text{and} \hspace{0.5 in} 
\rho=\frac{e^{\alpha R}(1-2 \alpha R)}{2 \kappa^{2}}.
\end{equation}
It is known that the equation of state parameter (EoS) is a relationship between pressure and energy density given by $\frac{p}{\rho}$. It is used to study the accelerated and decelerated phase of the universe. The universe exhibits various phases according to the various values of $\omega$.  In the case of $\omega=\frac{1}{3}$, the universe is governed by radiation. In the accelerated evolutionary phase, the quintessence phase is shown by $-1\leq \omega \leq 0 $ and the cosmological constant is shown by $\omega=-1$, i.e., $\Lambda$CDM model. The equation of EoS parameter in this case reads as $\omega= \frac{1}{2 R \alpha -1}$.

\begin{figure}[H]
\centering
\includegraphics[width= 8 cm]{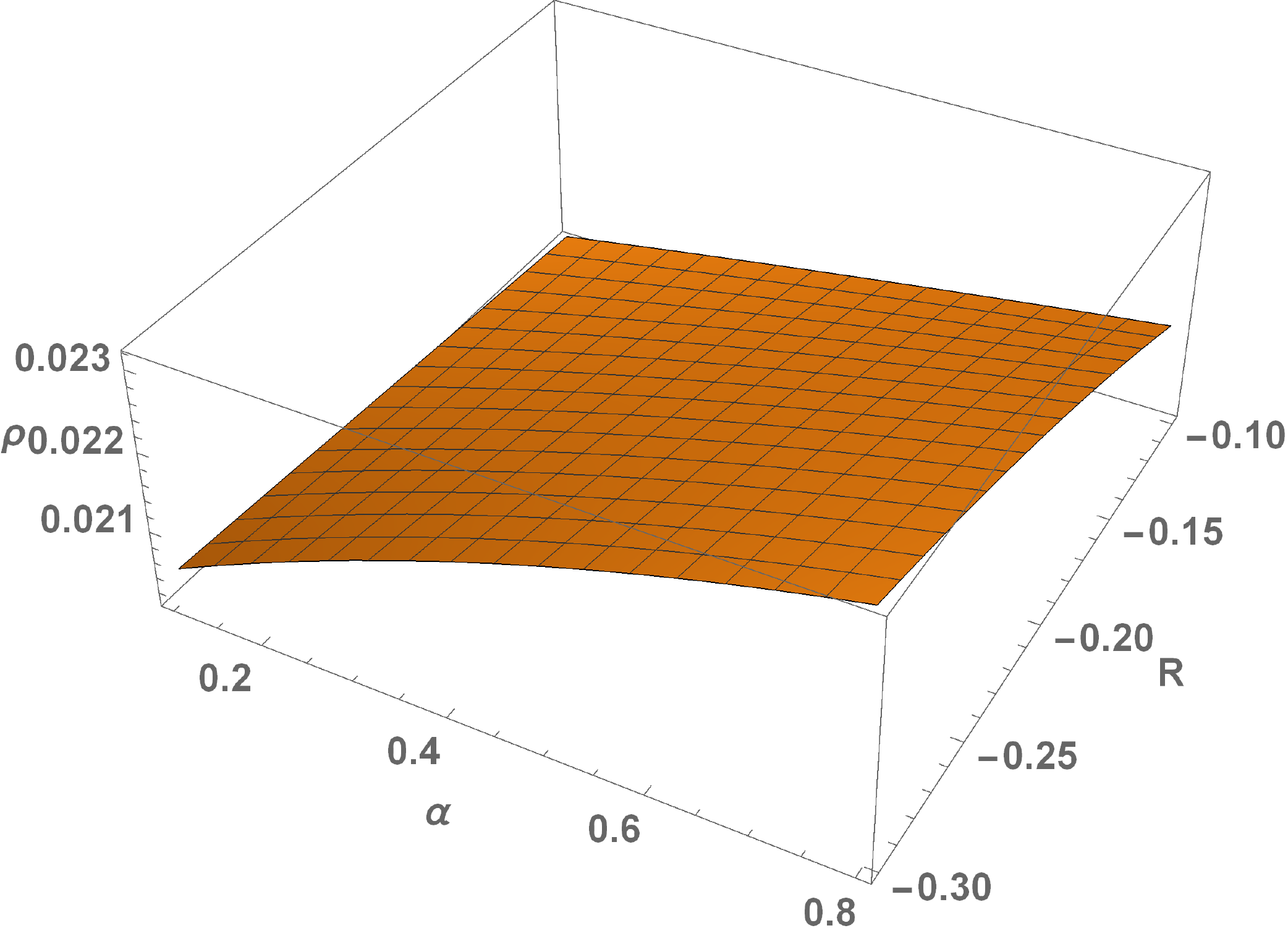} 
\caption{Density parameter for $F(R)=\exp(\alpha R)$ with $ 0.1 \leq \alpha \leq 0.8$ and $-0.3 \leq R\leq -0.1$.}
\label{Fig-1}
\end{figure}

\begin{figure}[H]
\centering
\includegraphics[width=8 cm]{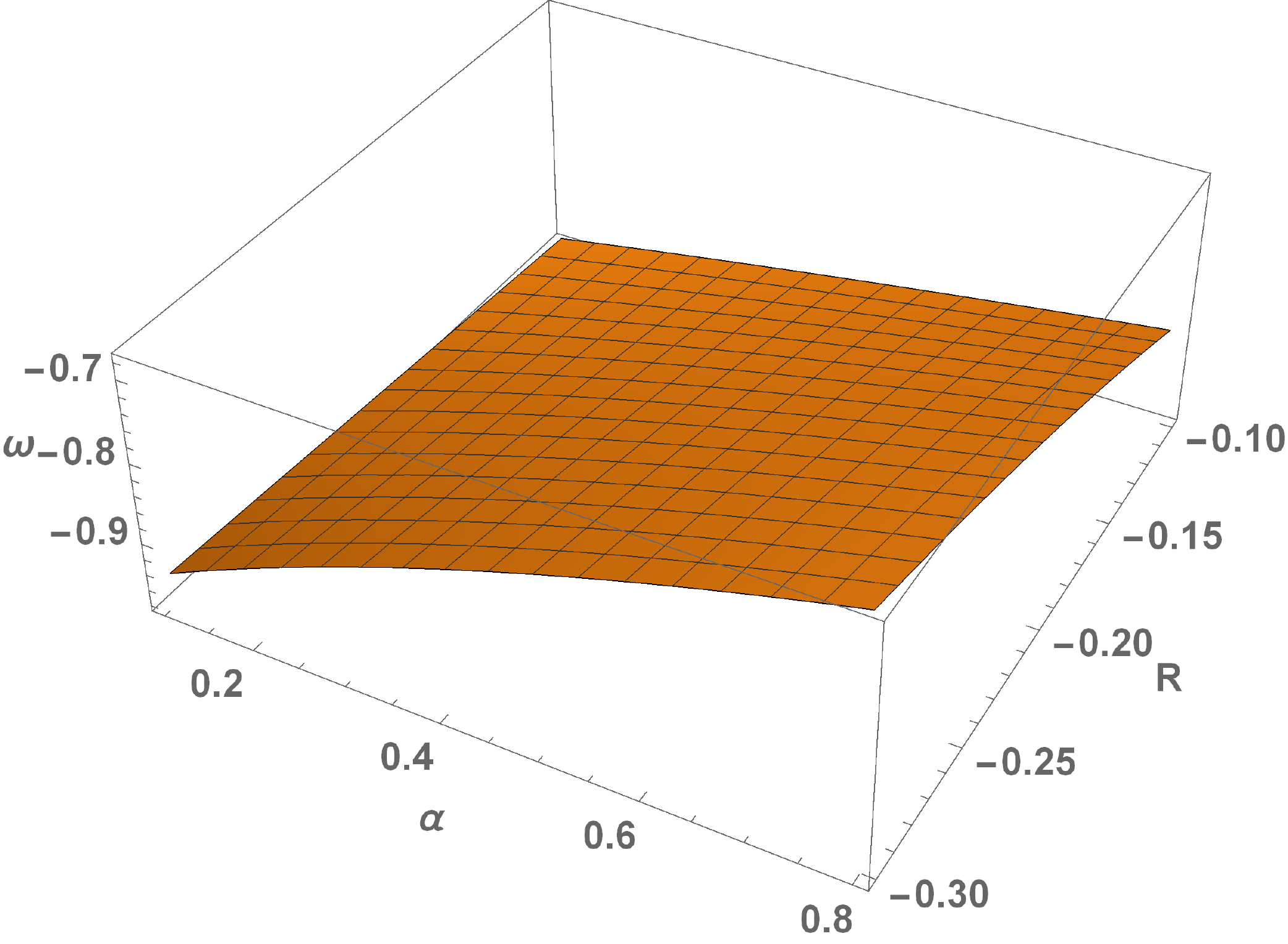}
\caption{EoS parameter for $F(R)=\exp(\alpha R)$ with $ 0.1 \leq \alpha \leq 0.8$ and $-0.3 \leq R\leq -0.1$.}
\label{Fig-2}
\end{figure}

The plots of density and EoS parameter are shown in Fig. \ref{Fig-1} and Fig. \ref{Fig-2}, which depict the positive behavior of the density parameter and EoS parameter nearly close to -1. Therefore,  the  EoS parameter is considered a suitable candidate for comparing our models with $\Lambda$CDM. In this section, the EoS parameter is useful in constraining the model parameters to study various energy conditions.
The previous section clearly states the conditions on $R$ to satisfy various energy conditions. So, Fig. \ref{fig-3} shows the behavior of NEC, DEC, and SEC. Since we know WEC is the combination of NEC and positive density. We can observe NEC, WEC, and DEC satisfy the conditions, whereas SEC violates them. According to recent observational studies \cite{Planck}, the violation of SEC depicts the accelerated expansion of the universe.

\begin{figure}[H]
\centering  
\subfigure[ NEC]{\includegraphics[width=0.49\linewidth]{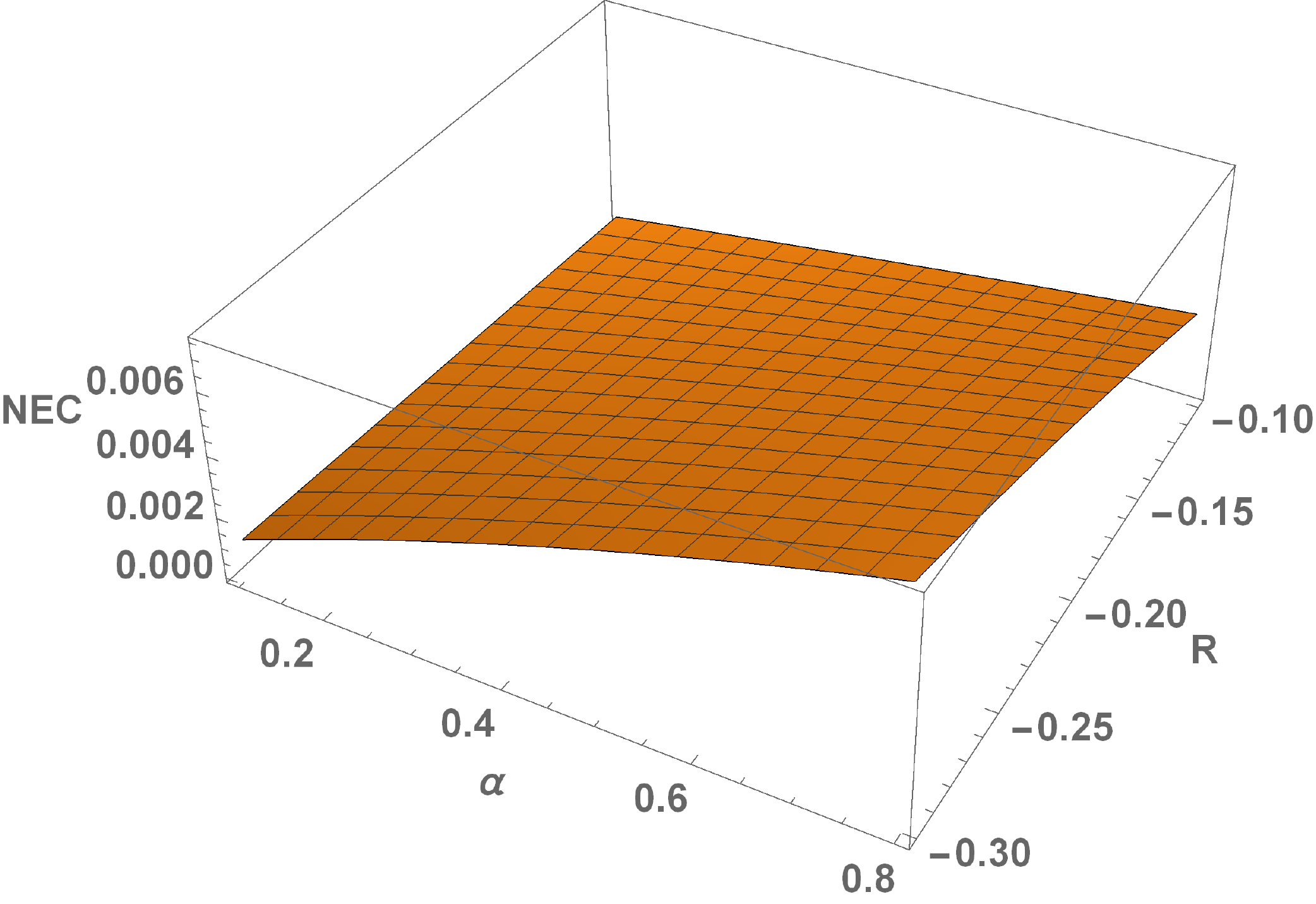}}
\subfigure[ DEC]{\includegraphics[width=0.49\linewidth]{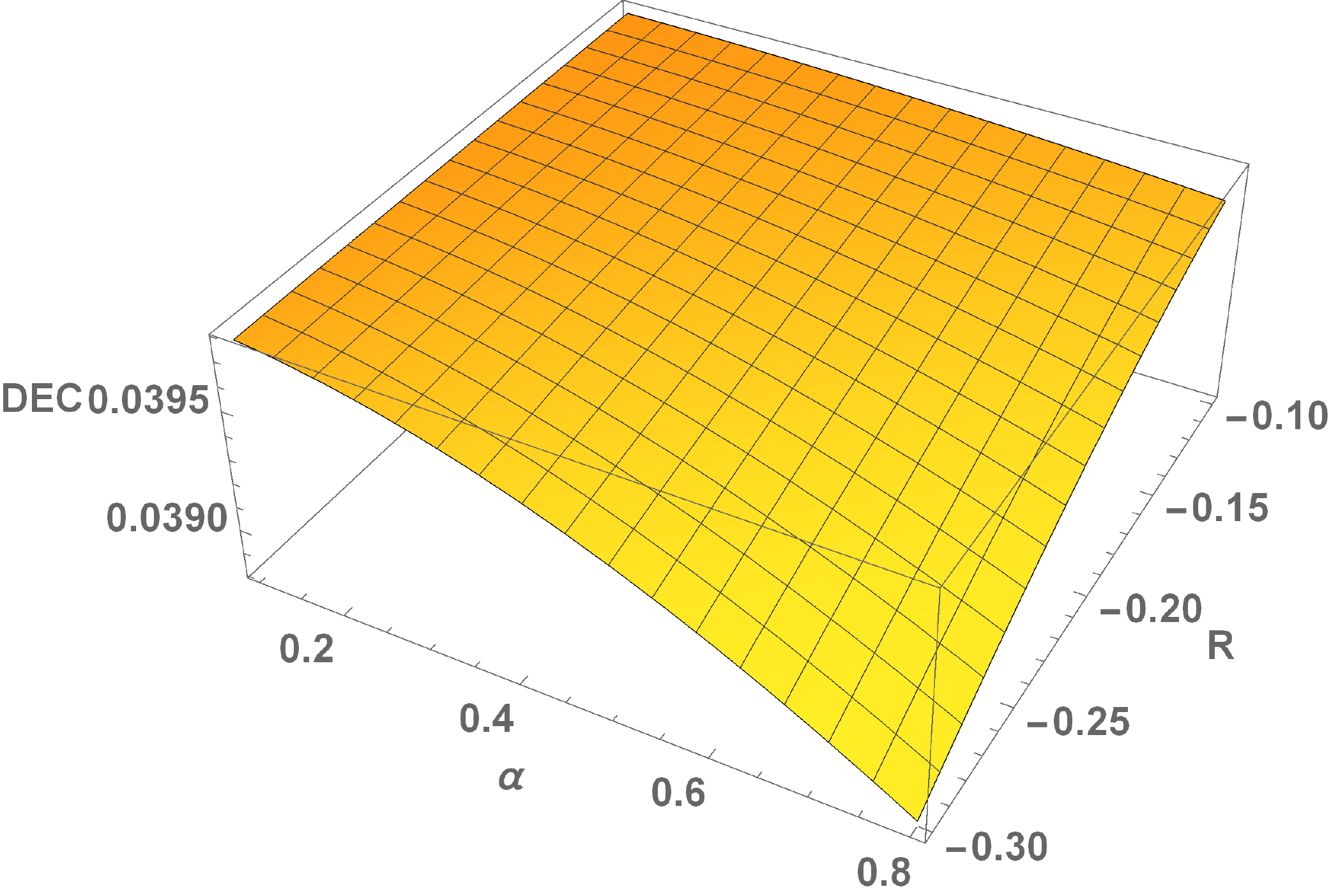}}
\subfigure[ SEC]{\includegraphics[width=0.49\linewidth]{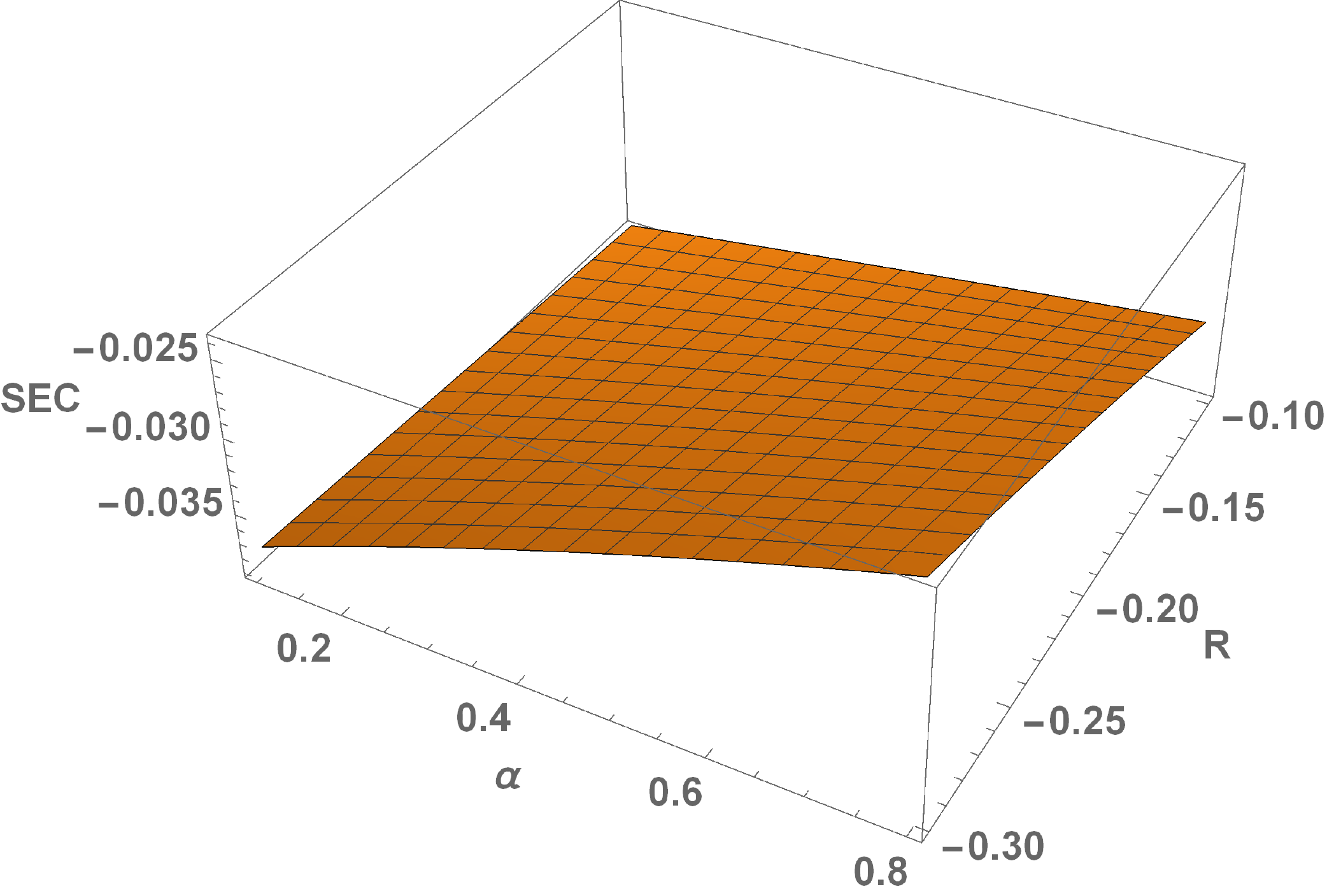}}

\caption{Energy conditions for $F(R)=\exp(\alpha R)$ with $ 0.1 \leq \alpha \leq 0.8$ and $-0.3 \leq R\leq -0.1$.}\label{fig-3}
\end{figure}
\textbf{Case II:} $F(R)=R- \beta \tanh(R) $, $\beta$ is a constant. S. 
Tsujikawa, Appleby and Battye also studied this type of $F(R)$ model \cite{Shinji}, \cite{Appleby}. 
The pressure and energy density for a perfect fluid continuum are as follows:
\begin{equation}
p=\frac{\beta \tanh(R)-R}{2 \kappa^{2}} \text{and} \hspace{0.5 in} 
\rho=\frac{2\beta R  \mathrm{\,sech}^{2}(R)-\beta \tanh(R)-R}{2 \kappa^{2}}.
\end{equation}

The equation of state parameter in this case reads as 
\begin{equation}
\omega= \frac{\beta  \tanh (R)-R}{-\beta  \tanh (R)+2 \beta  R \mathrm{\,sech}^2(R)-R}.
\end{equation}

\begin{figure}[H]
\centering
\includegraphics[width=8 cm]{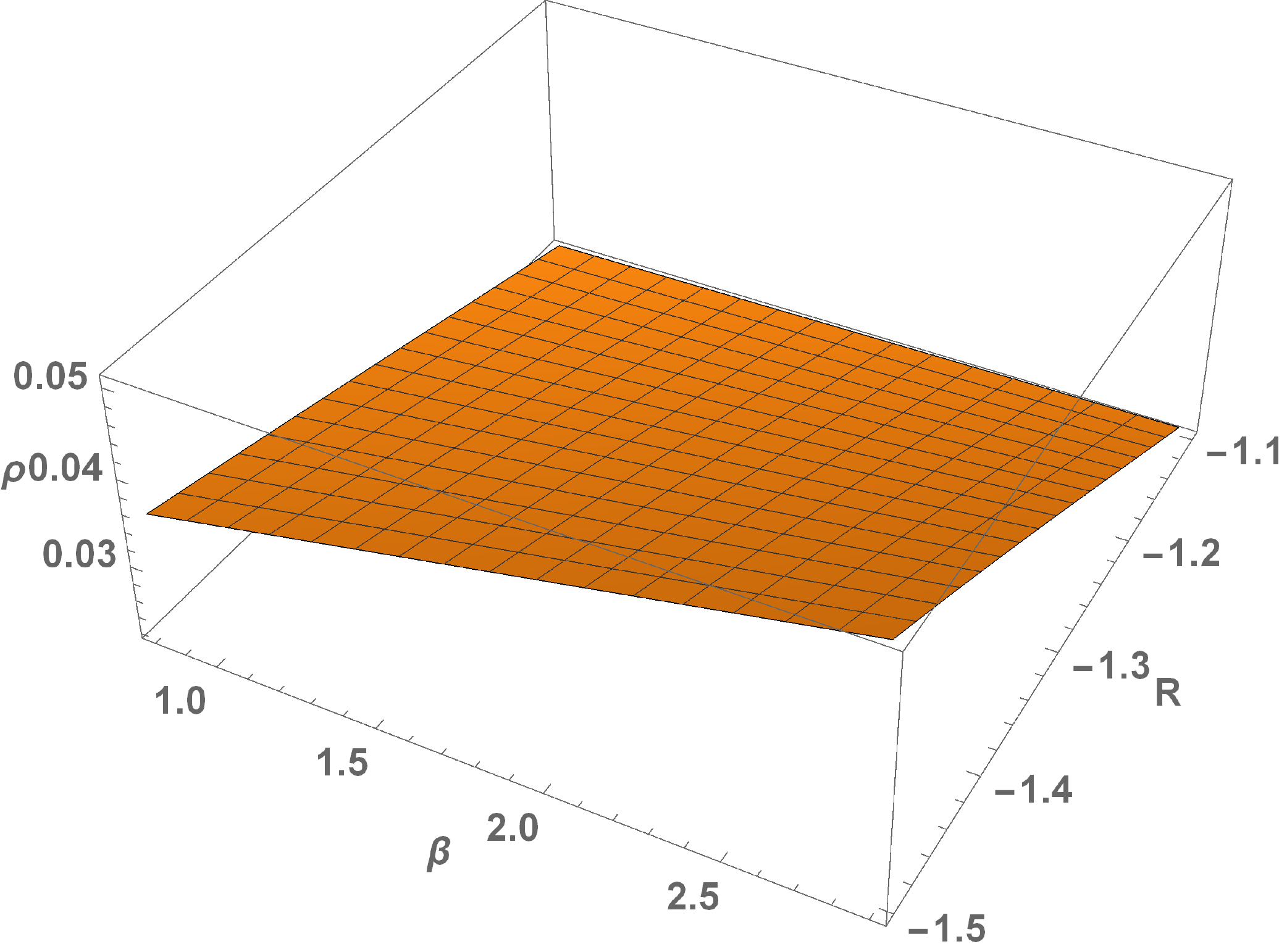} 
\caption{Density parameter for $F(R)=R- \beta \tanh(R) $  with $0.8\leq \beta \leq 2.9$ and  $ -1.5 \leq R\leq  -1.1$.}
\label{Fig-4}
\end{figure}

\begin{figure}[H]
\centering
\includegraphics[width=8 cm]{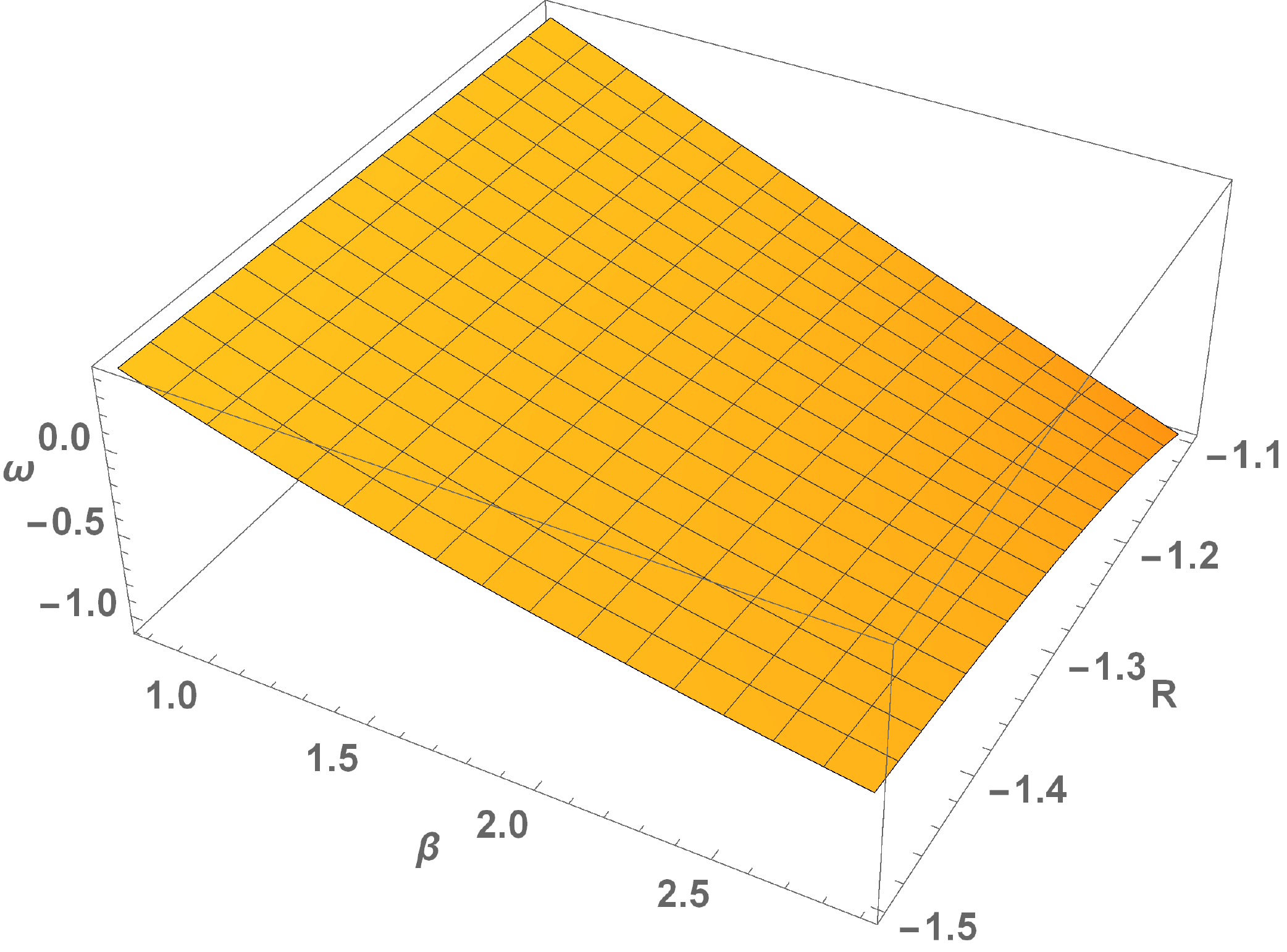}
\caption{EoS parameter for $F(R)=R- \beta \tanh(R) $ with $0.8\leq \beta \leq 2.9$ and  $ -1.5 \leq R\leq  -1.1$.}
\label{Fig-5}
\end{figure}

The plots of density and EoS parameter are shown in Fig. \ref{Fig-4}, and Fig. \ref{Fig-5}. The behavior of density is positive whereas, the EoS parameter shows the phase transition from positive to negative depicting a  transition from radiation dominated era to the accelerated phase of the universe. The EoS parameter is close to $-1$, which is consistent with the $\Lambda$CDM model. 

\begin{figure}[H]
\centering  
\subfigure[ NEC]{\includegraphics[width=0.49\linewidth]{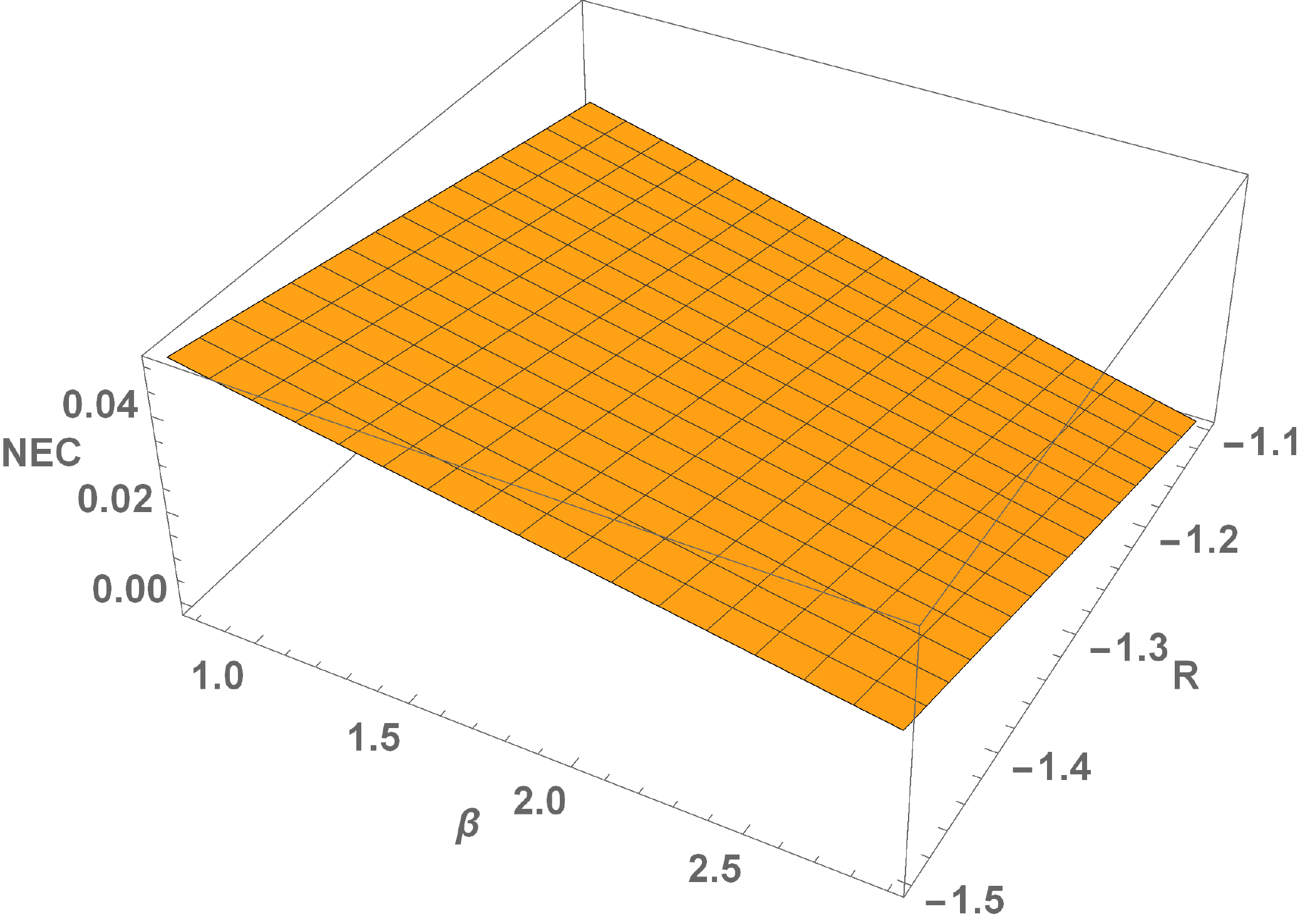}}
\subfigure[ DEC]{\includegraphics[width=0.49\linewidth]{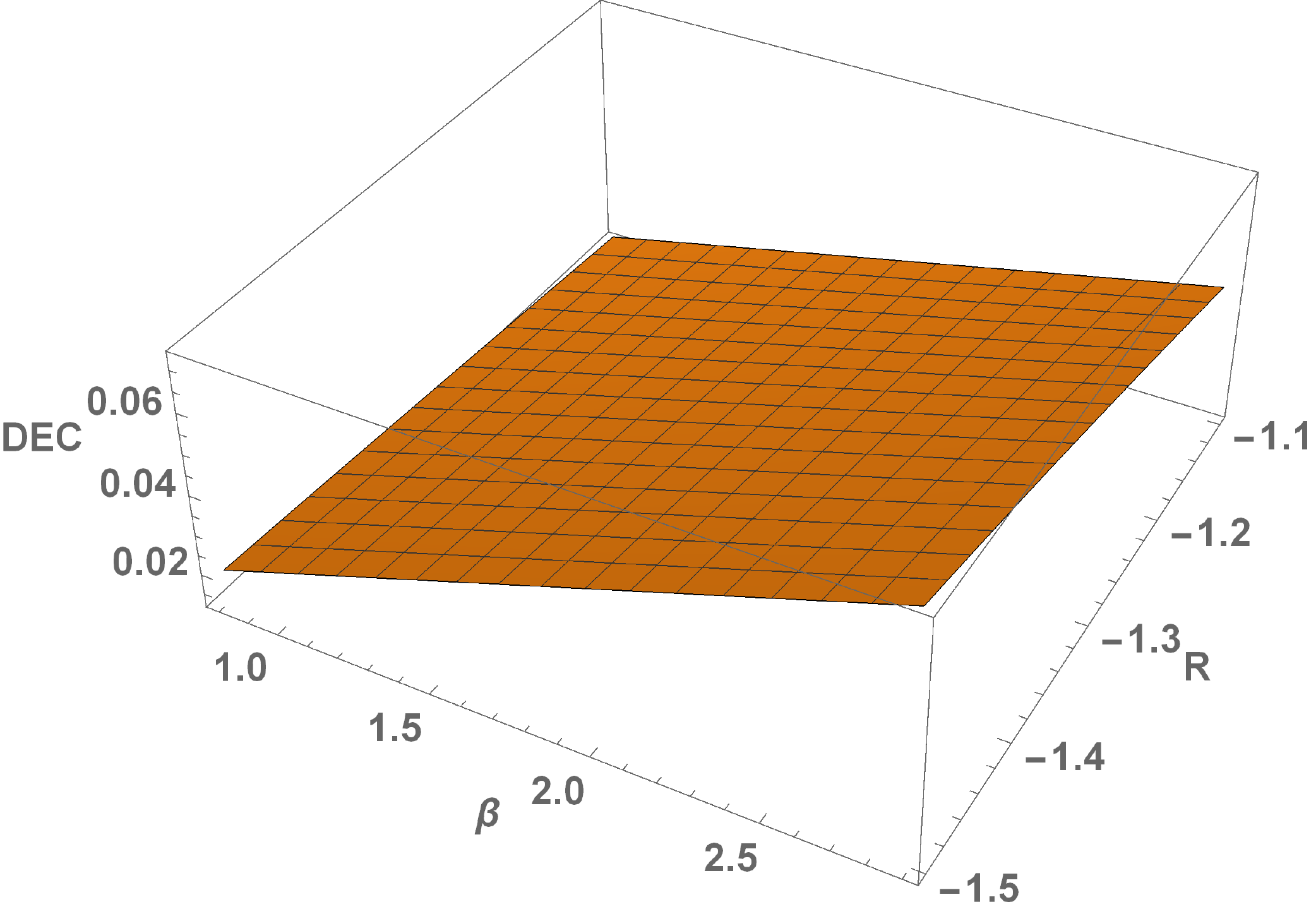}}
\subfigure[ SEC]{\includegraphics[width=0.49\linewidth]{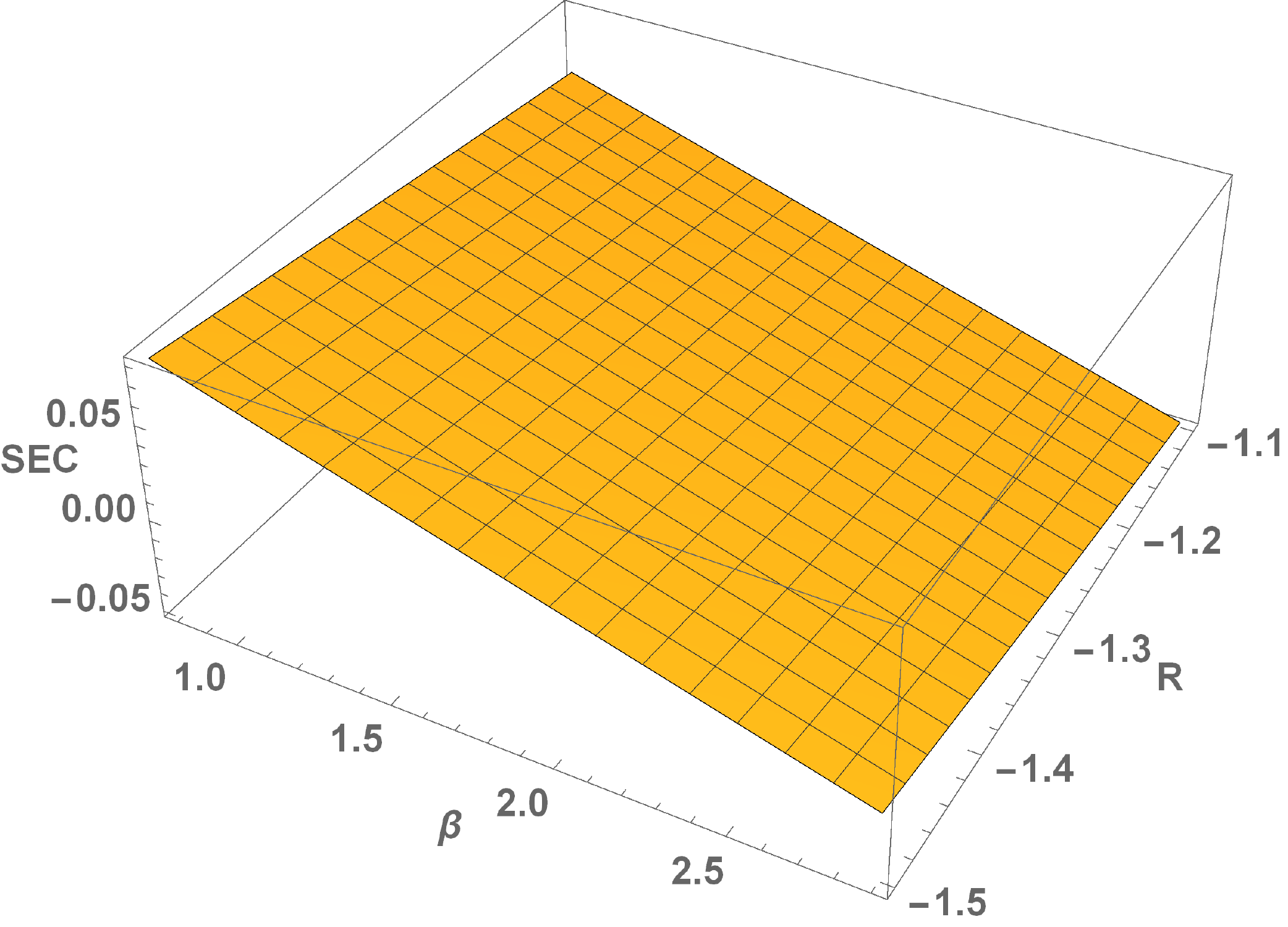}}

\caption{Energy conditions for $F(R)=R- \beta \tanh(R) $ with $0.8\leq \beta \leq 2.9$ and  $ -1.5 \leq R\leq  -1.1$.}\label{fig-6}
\end{figure}

Fig. \ref{fig-6} shows the behavior of NEC, DEC, and SEC. So, we can observe NEC, WEC, and DEC satisfy the conditions, whereas SEC shows the negative behavior, i.e., violating its condition. Therefore, according to recent observational studies, the SEC's violation depicts the universe's accelerated expansion.

\section{{Discussion}}
To understand the observed accelerated cosmic expansion, modified theories of gravity have gained considerable attention. One of the viability criteria of modified gravity theories is its compatibility with the causal and geodesic spacetime structure, which can be addressed through different energy conditions. In the present study, we investigate the strong, weak, null, and the dominant energy conditions for the modified $F(R)$-gravity theories under a geometric restriction of weakly Ricci symmetric type spacetimes. Nevertheless, the arbitrariness in choosing various functional forms of $F(R)$ used to constrain on physical grounds the several possible $F(R)$ gravity theories. 
In two separate $F(R)$ groups, i.e.,  $F(R)=e^{\alpha R}$ where $\alpha$ is constant and $F(R)=R-\beta \tanh(R)$, $\beta$ is a constant we explored different energy conditions.
 The conditions derived in Section \ref{conditions} are used to constrain the model parameters in two different $F(R)$ models. The model parameters must satisfy $R<0$, $\alpha \geq0$, 
and $\alpha \leq -\frac{1}{R}$ in Case I to satisfy different energy conditions, 
whereas in Case II, $R<0 $ and $\beta\geq 0 $ are provided.We observed that the null, weak and dominant energy conditions validate their requirements when $ 0.1 \leq \alpha \leq 0.8$, $-0.3 \leq R\leq -0.1$ and $0.8\leq \beta \leq 2.9$,  $ -1.5 \leq R\leq  -1.1$ for the two cases respectively. However, on the other, the strong energy conditions show negative behavior indicating the violation of the conditions. We conclude that the matter contained in this current setting is a perfect fluid with vanishing vorticity.

\section*{Acknowledgments}
Avik De and Tee-How Loo are supported by the grant FRGS/1/2019/STG06/UM/02/6. 
S. A. acknowledges CSIR, Govt. of India, New Delhi, for awarding Junior Research Fellowship. 
PKS   acknowledges CSIR, 
New Delhi, India for financial support to carry out the Research project [No. 03(1454)/19/EMR-II Dt 02/08/2019]. 
We are very much grateful to the honorable referee and the editor for the illuminating suggestions that 
have significantly improved our work in terms of research quality and presentation.

\end{document}